\documentclass[11pt]{article}
\usepackage{fullpage}
\usepackage{url}
\usepackage{xspace}

\usepackage{graphics}
\usepackage[dvips]{epsfig}

\usepackage{amsmath}
\usepackage{amssymb}
\usepackage{amsfonts}
\usepackage{graphicx}
\usepackage[ruled,vlined]{algorithm2e}

\newtheorem{theorem}{Theorem}[section]

\newtheorem{proposition}{Proposition}[section]

\newcommand{\qed}{\hfill $\Box$ \bigbreak}
\newenvironment{proof}{\noindent {\bf Proof.}}{\qed}

\newcommand{\remove}[1]{}

\begin{document}

\baselineskip  0.19in 
\parskip     0.05in 
\parindent   0.3in 

\title{{\bf Deterministic Treasure Hunt and Rendezvous in Arbitrary Connected Graphs
 }}
\date{}
\newcommand{\inst}[1]{$^{#1}$}

\author{
Debasish Pattanayak\inst{1},
Andrzej Pelc\inst{1}$^,$\footnote{Partially supported by NSERC discovery grant 2018-03899 and by the Research Chair in Distributed Computing at the Universit\'e du Qu\'{e}bec en Outaouais.}\\
\inst{1} Universit\'{e} du Qu\'{e}bec en Outaouais, Gatineau, Canada.\\
E-mails: \url{ drdebmath@gmail.com}, \url{ pelc@uqo.ca}\\
}

\maketitle

\begin{abstract}
Treasure hunt and rendezvous are fundamental tasks performed by mobile agents in graphs. 
In treasure hunt, an agent has to find an inert target (called treasure) situated at an unknown node of the graph. 
In rendezvous, two agents, initially located at distinct nodes of the graph,
traverse its edges in synchronous rounds and
have to meet at some node. 
We assume that the graph is connected (otherwise none of these tasks is feasible) and consider deterministic treasure hunt and rendezvous algorithms.
The time of a treasure hunt algorithm is the
worst-case number
of edge traversals performed by the agent until the treasure is found.
The time of a rendezvous algorithm is the worst-case number of rounds
since the wakeup of the earlier agent until the meeting. 

To the best of our knowledge, all known treasure hunt and rendezvous algorithms rely on the assumption that degrees of all nodes are finite, even when the graph itself may be infinite. In the present paper we remove this assumption for the first time, and consider both above tasks in arbitrary connected graphs whose nodes can have either finite or countably infinite degrees. Our main result is the first universal treasure hunt algorithm working for arbitrary connected graphs. We prove that the time of this algorithm has optimal order of magnitude among all possible treasure hunt algorithms working for arbitrary connected graphs. 
 As a consequence of this result we obtain the first universal rendezvous algorithm working for arbitrary connected graphs. The time of this algorithm is polynomial in a lower bound holding in many graphs, in particular in the tree all of whose degrees are infinite. 
\end{abstract}



\textbf{keywords.} treasure hunt, rendezvous, deterministic algorithm, mobile agent, time






\section{Introduction}

\subsection{The background}
Treasure hunt and rendezvous are related fundamental tasks performed by mobile entities, called agents, in graphs. 
In treasure hunt, an agent has to find an inert target (called treasure) located at an unknown node of the graph. 
In rendezvous, two agents, initially located at distinct nodes of the graph,
traverse its edges in synchronous rounds and
have to meet at some node. In both cases we assume that the graph is connected, otherwise these tasks cannot be accomplished if the initial location of the agent and the target (in case of treasure hunt) and the initial locations of the agents (in case of rendezvous) are in different components.
Both these tasks have numerous applications. Mobile entities may model software agents navigating in a computer network. In this application, an example of the task of treasure hunt is finding some data hidden in an unknown node of the network. Meeting of software agents may be important to exchange data previously found by each of them. In robotics applications, agents may model mobile robots navigating in a network of corridors forming a graph. Then an example of the task of treasure hunt is finding a person stranded in a building after an accident, and meeting of mobile robots may be necessary to exchange information about parts of the environment already visited by each of them.

To the best of our knowledge, all known treasure hunt and rendezvous algorithms rely on the assumption that degrees of all nodes are finite, even when the graph itself may be infinite. In the present paper we remove this assumption for the first time, and consider both above tasks in arbitrary connected graphs whose nodes can have either finite or countably infinite degrees. It turns out that removing the finite degree assumption significantly changes the problem of finding efficient treasure hunt and rendezvous algorithms.

From a practical point of view, our approach has the advantage of improving complexity of treasure hunt and rendezvous for graphs with nodes of finite but very large degrees. Such graphs occur in social networks, or in dynamic webpages subject to crawling by software agents. In  subsection \ref{Our results} we explicitly mention such improvements.

\subsection
{The model and the problems}
We consider arbitrary undirected connected\footnote{A graph is connected, if every pair of nodes is joined by at least one finite path.} graphs (called {\em graphs} in the sequel) whose nodes are either finite or countably infinite\footnote{The restriction to finite or countably infinite degrees is necessary: in graphs with some nodes of uncountable degrees neither treasure hunt nor rendezvous can be executed in worst-case finite time.}. They may have self-loops and/or multiple edges. Port numbers at any node of finite degree $d$ are arbitrarily numbered $1, \dots,d$. Port numbers at any node of infinite degree are arbitrarily numbered by all distinct positive integers. We do not make any assumptions on the existence of node labels (i.e., we assume that graphs are anonymous).
Nodes cannot be marked by the agents in any way.
Agents have unbounded memory: from the computational perspective they are modelled as Turing machines. 

$\bullet$ Treasure hunt.\\
An instance of the treasure hunt problem is an ordered pair $(b,t)$, where $b$, called the {\em base}, is the starting node of the agent, and $t$ is a node which is the location of the treasure.
The agent has no a priori knowledge whatsoever.
At each step, the agent chooses a port at the current node and traverses the corresponding edge.
When an agent enters a node, it learns the entry port number.  
The time of a treasure hunt algorithm is the
worst-case number
of edge traversals performed by the agent until the treasure is found.
A {\em universal treasure hunt algorithm} is an algorithm that solves the treasure hunt problem for every instance on every connected graph.

$\bullet$ Rendezvous.\\
Agents have distinct labels drawn from the set $\{1,\dots, L\}$.
An instance of the rendezvous problem is a set $\{(v_1,\ell_1),(v_2,\ell_2)\}$, where $v_i$ are the starting nodes of the agents and $\ell_i$ are their labels.
Each agent knows a priori only its own label.
Agents execute the same deterministic algorithm with a parameter that is the agent's label.
Agents move in synchronous rounds.
In each round, each agent either stays idle at the current node, or chooses a port at the current node
and traverses the corresponding edge.
If agents cross each other, traversing simultaneously the same edge in opposite directions, they do not even notice this fact.
When an agent enters a node, it learns the entry port number. 
Agents are woken up in possibly different rounds by an adversary. 
The time of a rendezvous algorithm is the worst-case number of rounds
since the wakeup of the earlier agent until the meeting. 
A {\em universal rendezvous algorithm} is an algorithm that solves the rendezvous problem for every instance on every connected graph.

\subsection
{Terminology}
A (finite) path from node $u$ to node $v$ in a graph $G$ is a sequence $(p_1,\dots,p_\delta)$ of port numbers, such that taking port number $p_1$ at node $u$ and then taking port numbers $p_2,\dots,p_\delta$ at subsequent nodes, reaches node $v$.
Consider a path $\pi = (p_1, p_2, \dots, p_\delta)$ starting at a node $v_0$ of a graph. Let $v_0,v_1, v_2, \dots, v_\delta$ be the nodes such that
port $p_i$ at $v_{i-1}$ leads to $v_i$, for  $i=1,\dots,\delta$. We define the {\em reverse} of the path $\pi$ as the path $rev(\pi) = (q_\delta, q_{\delta-1}, \dots, q_1)$ such that port $q_i$ at node $v_i$ leads to $v_{i-1}$, for  $i=1,\dots,\delta$. 
The length $\delta$ of path $\pi$ is denoted by $|\pi|$.
Since an agent learns the entry port upon visiting a node, it learns path 
$rev(\pi)$ after traversing path $\pi$.

Consider an ordered couple $(u,v)$ of nodes and a path $\pi=(p_1,\dots,p_\delta)$ from $u$ to $v$. The {\em type} of $\pi$, denoted $\tau(\pi)$,
is the ordered couple $(m,\delta)$, where $m=\max(p_1,\dots,p_\delta)$. 
The {\em value} $val(x,y)$ of any pair $(x,y)$ of natural numbers is defined as $y2^yx^y$.
The {\em character} of a couple $(u,v)$ of nodes, denoted $c(u,v)$, is  defined as follows.
Let $minval$ be the smallest value of any type of a path from $u$ to $v$. 
$c(u,v)$ is  the lexicographically smallest among all types of value $minval$. Finally, the {\em weight} of a couple $(u,v)$ of nodes, denoted $w(u,v)$, is defined as the value of $c(u,v)$.
Notice that $w(u,v)$ is not necessarily equal to $w(v,u)$.
As an example, consider a pair $(u,v)$ of nodes that has four paths from $u$ to $v$: $\pi_1=(2,1,2,1)$, $\pi_2=(3,10)$, $\pi_3=(4,3)$, and $\pi_4=(64)$. We have $\tau(\pi_1)=(2,4)$, $\tau(\pi_2)=(10,2)$, $\tau(\pi_3)=(4,2)$, $\tau(\pi_4)=(64,1)$. The value of $(2,4)$ is 1024, the value of 
$(10,2)$ is 800, the value of $(4,2)$ is 128, and the value of $(64,1)$ is 128. Hence the types with the smallest value 128, which is $minval$ of $(u,v)$, are 
$\tau(\pi_3)=(4,2)$ and $\tau(\pi_4)=(64,1)$. Thus the lexicographically smallest among them, i.e., $(4,2)$ is $c(u,v)$, and the weight $w(u,v)$ is the value of this character, i.e., 128. The notion of weight of a couple of nodes is crucial to the formulation of our results.


\subsection
{Our results}\label{Our results}
 Our main result is the first universal treasure hunt algorithm. We prove that the time of this algorithm has optimal order of magnitude among all possible universal treasure hunt algorithms. More precisely, the complexity of the algorithm for an instance $(b,t)$ of treasure hunt in any graph $G$ is $O(w(b,t))$.We show that this complexity is optimal.


As an application, we obtain the first universal rendezvous algorithm.
This algorithm is based on our universal treasure hunt algorithm and on the rendezvous algorithm from \cite{BP}. 
For any instance
 of rendezvous in any connected graph $G$, where $v_1$ and $v_2$ are the initial positions of the agents, and agents have distinct labels belonging to the set $\{1,\dots,L\}$,
our algorithm has time polynomial in $\max(W(v_1,v_2),\log L)$, where \sloppy$W(v_1,v_2)=\max(w(v_1,v_2),w(v_2,v_1))$. On the other hand, we show that $\Omega(\max(W(v_1,v_2),\log L))$ is a lower bound on the time of any universal rendezvous algorithm, which holds, e.g., in the infinite tree $T$ with all nodes of infinite degrees.

It should be noted that the rendezvous algorithm from \cite{BP} worked only for connected graphs all of whose nodes have finite degrees (thus the title of \cite{BP}
``How to meet at a node of any connected graph'' was somewhat overstated). The analysis of this algorithm used the assumption of degree finiteness in a crucial way.
Indeed, for any node $v$ and any positive integer $r$, the authors defined $P(v,r)$ to be the number of paths of length $r$ in the graph, starting at node $v$.
For any instance of the rendezvous problem, where agents start at nodes $v_1$ and $v_2$ at distance $D$, the authors  defined $P(v_1,v_2,D)=\max(P(v_1,D),P(v_2,D))$ and showed that the time of their algorithm is polynomial in $\max(P(v_1,v_2,D),D ,\log L)$, while
 $\Omega(\max(P(v_1,v_2,D),D ,\log L))$ is a lower bound on rendezvous time in trees. Of course, in case of infinite degrees,  $P(v_1,v_2,D)$ may be infinite, so this result becomes meaningless. 
 
 While our rendezvous algorithm is designed to work even for graphs with nodes of countably infinite degrees, it also offers significant improvements over \cite{BP} for many instances in graphs with finite but very large node degrees. 
For instances of rendezvous, such that $L$ and $W(v_1,v_2)$ are constant,  in graphs with nodes of finite but unbounded degrees, our algorithm has constant time, while the time of the algorithm from  \cite{BP} is unbounded.

\subsection
{Related work}
Treasure hunt and rendezvous are related problems intensely studied in the last few decades.
The task of treasure hunt has been investigated in the line \cite{BCD,HIKL}, in the plane \cite{BCR,La} and in terrains with obstacles \cite{LS}. In \cite{PY}, the authors studied the minimum amount of information required by an agent in order to perform efficient treasure hunt in geometric terrains.  Treasure hunt in finite anonymous graphs was studied in \cite{TSZ07,Xin} with the goal of minimizing time.
Treasure hunt in connected, possibly infinite, graphs with labeled nodes of finite degrees was investigated in~\cite{ABRS,BDLP}.
An optimal randomized algorithm for treasure hunt in a line is presented in \cite{KRT}.
In \cite{KKKS}, the authors considered expected competitive ratio for randomized on-line treasure hunt.

The problem of rendezvous in graphs was investigated under randomized and deterministic scenarios.
It is usually assumed that the nodes do not have distinct identities.
A survey of  randomized treasure hunt and rendezvous in various models  can be found in
\cite{alpern02b}. Deterministic rendezvous in graphs was surveyed in \cite{Pe2}.
Gathering more than two agents in the plane was studied, e.g., in \cite{fpsw}. 
The related problem of approach in the plane, where agents have to get at distance 1, was studied in \cite{BBDDP}.

Most of the literature on rendezvous considered finite graphs and assumed the synchronous scenario, where
agents move in rounds. 
In \cite{MP} the authors studied tradeoffs between the amount of information a priori available to the agents and time of  treasure hunt and of rendezvous. 
 In \cite{TSZ07}, the authors presented rendezvous algorithms with time polynomial in the size of the graph and the length of agent's labels.
Gathering many agents in the presence of Byzantine agents was discussed in \cite{BDL}.
The amount of memory required by the agents to achieve deterministic rendezvous was studied in  \cite{CKP} for general graphs.

Fewer papers were devoted to synchronous rendezvous in infinite graphs. In \cite{CCGKM}, the authors considered rendezvous in infinite trees and grids but they used the strong assumption that the agents know their location in the environment. In \cite{BP2}, the authors considered rendezvous in infinite trees but in the case of unoriented trees they limited attention to regular trees. The paper \cite{BP}, most closely related to the present paper, was discussed previously. In particular, our rendezvous algorithm is partially based on algorithm {\tt RV} from \cite{BP}.  In all the above papers it was assumed that all nodes in graphs under consideration have finite degrees.

Several authors investigated asynchronous rendezvous in the plane \cite{CFPS,fpsw} and in graphs \cite{BCGIL,DPV}.

\section{Universal Treasure Hunt}

In this section, we describe and analyse a universal algorithm for the task of treasure hunt in arbitrary connected graphs. 
The main idea is to try to traverse all possible paths starting at the base (some paths may not exist) in non-decreasing order of values of their types. Types of the same value and paths of the same type could be ordered arbitrarily. For convenience, we choose the lexicographic order in both cases. It turns out that the non-decreasing order of values of path types guarantees  optimal complexity of universal treasure hunt.

Now we give a detailed description of the algorithm. First, we describe
two auxiliary procedures. Procedure Traverse($\pi$), for a given path $\pi = (p_1, p_2, \dots, p_\delta)$, starts and ends at the base node of the agent. Its aim is to follow consecutive ports $p_i$, as long as they are available at subsequently visited nodes, and then take the reverse path, in order to get back to the base. For example, if 
$\pi=(5,2,3,5,4,1)$ but the fourth visited node has degree 3, so it has no port 5, then the agent  backtracks from this node to the base.
Procedure Traverse($\pi$) lasts at most $2|\pi|$ steps.

The next procedure is Procedure Paths$(m, \delta)$, for any integers $m\geq 2$ and $\delta>0$. 
It starts and ends at the base node of the agent. Its aim is to  traverse 
all the paths of type $(m, \delta)$ in lexicographic order, using Procedure Traverse,  and get back to the base.

Let $\Pi$ be the set of all paths of type $(m, \delta)$.
To construct $\Pi$, take all possible sequences of length $\delta$ of port numbers $(p_1, p_2, \dots, p_\delta)$, such that $p_i \leq m$ for all $i$. Then for each such sequence, find $\kappa = \max(p_1, p_2, \dots, p_\delta)$. Add the sequence to $\Pi$ if $\kappa = m$. The number of such sequences is, precisely, $|\Pi| = m^\delta - (m-1)^\delta$.
In Procedure Paths$(m, \delta)$, the agent stores the set $\Pi$ in its memory, and tries to traverse all paths in $\Pi$ in lexicographic order (some paths may not exist), using
Procedure Traverse($\pi$) to traverse each path.

Now we are ready to describe our main algorithm, called {\tt UTH} (for Universal Treasure Hunt). 
It works in phases corresponding to positive integers starting at 2. For any integer $j\geq 2$, the phase $j$ works as follows.
We determine all types $(m, \delta)$, for any integers $m\geq 2$ and $\delta>0$,  such that  $val(m, \delta)=j$,  and sort these types lexicographically.
Then the agent executes Procedure Paths$(m, \delta)$ for each type $(m, \delta)$ in the sorted order.
This completes phase $j$ and the next phase starts. 
Hence, Algorithm {\tt UTH} yields an ordering $(\pi_1^*,\pi_2^*,\dots)$ of all finite paths starting at the base of the agent.
Below is the pseudocode of this algorithm.
The execution of the algorithm stops as soon as the treasure is found.

\begin{algorithm}[H]
  \SetAlgoLined
  \SetKwBlock{Repeat}{repeat forever}{}
  $j\gets 1$\;
  \Repeat{
    $j\gets j + 1$\;
    Let $\Sigma$ be the set of all pairs ($x,y$) of integers, such that $j = y\cdot 2^y \cdot x ^ y$, $x\geq 2$, and $y>0$\;
    Let $\sigma _1,\dots, \sigma_s$ be the sequence of pairs from  $\Sigma$ ordered lexicographically\;
    \For{$i=1$ to $s$}{
      Execute Paths$(\sigma_i)$\;
    }
  }
  \caption{{\tt UTH}}\label{alg:uth}
\end{algorithm}

Now, we analyze the complexity of Algorithm {\tt UTH}.
First we need to discuss the parameters in terms of which this complexity is expressed. Observe that global parameters, such as the number of nodes or  edges in the graph, its diameter or its maximum degree, which are widely used in the literature concerning finite graphs, cannot be used in our scenario because, in case of infinite graphs, all these parameters could have infinite value. Another possibility, adopted, e.g., in  \cite{ ABRS,BDLP} and applicable even for infinite graphs but assuming that degrees of all nodes are finite, would be to choose a parameter based on the size of the ball of radius $\delta$ in the graph, where $\delta$ is the distance from the starting node of the agent to the treasure. However, since we are dealing with graphs whose nodes may have infinite degree, this approach is also impossible because even the number of nodes adjacent to the starting node can be infinite. Hence a different type of parameter is needed.

We express the complexity of Algorithm {\tt UTH} in terms of the weight $w(b,t)$ of an instance $(b,t)$ of the treasure hunt problem.
If $c(b,t)=(m,\delta)$ then $w(b,t)$ is the value of $c(b,t)$, i.e., $w(b,t) = \delta2^\delta m^\delta$.
The reason of using $w(b,t)$ as a parameter is the following.
The amount of time required to explore all paths of length $\delta$ with maximum port number $m$, starting from $b$,  is $2\delta m^\delta$.
Any exploration of all these paths includes exploration of all paths of length {\em at most} $\delta$ with maximum port number $m$, starting from $b$. The weight of an instance is defined in such a way that, for a given maximum port number, even 
if we compute the cumulative weight for all instances  corresponding to shorter paths, it is still smaller than the weight of the instance for the longest path. This is accomplished by adding a factor $2^{\delta-1}$ to the time of exploration.
\begin{theorem}
\label{thm:uth}
Algorithm {\tt UTH} finds the treasure in every connected graph, in time $O(w(b,t))$, where $w(b,t)$ is the weight of the instance $(b,t)$ of the treasure hunt problem.
\end{theorem}

\begin{proof}
Consider an instance $(b,t)$ of the treasure hunt problem, and let $c(b,t)=(m,\delta)$. Then $w(b,t) = \delta2^\delta m^\delta$.
Since the agent traverses paths in non-decreasing order of values of their types,  and it stops when it visits the node $t$, the total number of phases of the algorithm is $w(b,t) - 1$.
The total number of paths of type $(x, y)$ is precisely $x^y - (x-1)^y$. The time of traversing all paths of type $(x, y)$ is at most $2y (x^y - (x-1)^y)$.
By the time the agent starts traversing paths of type $(x,y)$, it has already traversed all paths of type $(x', y)$ where $x' < x$.
The total time of traversing all paths $\pi$ of length $y$, such that the value of $\tau(\pi)$ is at most $w(b,t) $, is
$$ \sum_{i = 1}^x 2y(i^y - (i-1)^y) = 2y x^y =  \frac{2(y2^yx^y)}{2^y} \leq \frac{2w(b,t) }{2^{y}} = \frac{w(b,t) }{2^{y-1}}.
$$

Consider all types $(x,y)$ with value at most $w(b,t) $.
Let $h$
be the largest integer $y$ such that $val(2,y)= y2^{2y} \leq w(b,t) $.
The largest  path length $y$ of such a type is at most $h$.
Hence, the total time of the first $w(b,t) $ phases of the algorithm is at most $\sum_{y=1}^h \frac{w(b,t) }{2^{y-1}} \leq 2w(b,t) $.
It follows that the time of the algorithm is at most $2w(b,t)  \in O(w(b,t) )$.
\end{proof}

Since Theorem \ref{thm:uth} asserts that Algorithm {\tt UTH} works in time $O(w(b,t) )$ for every connected graph, in order to prove that this algorithm has optimal time complexity among universal treasure hunt algorithms, it is enough to show one connected graph with the property that, for any universal treasure hunt algorithm, there exist instances of the treasure hunt problem of arbitrarily large weights, for which this algorithm has time at least proportional to these weights.

The following proposition shows that Algorithm {\tt UTH} has optimal time complexity among universal treasure hunt algorithms.

\begin{proposition}\label{lb}
For any universal treasure hunt algorithm $\mathcal{A}$ executed on the infinite tree $T$ with all nodes of infinite degree, there exist instances $(b,t)$ of the treasure hunt problem, with arbitrarily large weight,
such that the time required to find the treasure for the instance $(b,t)$, using algorithm $\mathcal{A}$, is at least  $w(b,t)/4$.
\end{proposition}

\begin{proof}
Let $i$ be any positive integer.
For a given base node $b$ of $T$, consider instances $(b,t_1)$, $(b,t_2)$,\dots, $(b,t_{2i})$, such that $t_j$ is adjacent to $b$ and the port number at $b$ corresponding to the edge $\{b,t_j\}$ is $j$. For all $i+1 \leq j \leq 2i$, we have $w(b,t_j)=2j$ and hence $2(i+1)\leq w(b,t_j) \leq 4i$.
Let $t_r$ be the node among $t_{i+1}$,\dots, $t_{2i}$ that algorithm $\mathcal{A}$ visits last. The adversary puts the treasure in node $t_r$.
Hence the time of solving the treasure hunt problem by algorithm $\mathcal{A}$ for instance $(b,t_r)$ is at least $2i-1\geq i \geq w(b,t_r)/4$.
\end{proof}

Although the example of the tree $T$ from Proposition \ref{lb} and the chosen targets adjacent to the  base $b$ are enough to show optimality of 
Algorithm {\tt UTH}, many more such examples are possible to construct, so that Proposition \ref{lb} still holds for an appropriately changed constant replacing 1/4. For example, we could consider targets at a given constant distance from the base. Also, modifications of the tree $T$ by adding  edges between nodes at larger levels than that of the chosen targets, would not change the argument.

\section{Universal Rendezvous}

In this section we present the first rendezvous algorithm working for arbitrary connected graphs. It is based on our universal treasure hunt algorithm presented in the previous section, and on the rendezvous algorithm from \cite{BP} that worked for connected graphs with all nodes of finite degrees.

We first summarize the original Algorithm {\tt RV} from \cite{BP}. It has three parts.
Suppose that the starting node of the agent is $v$.
\begin{enumerate}
\item
Definition of the tape.\\
Given the label $\ell$ of the agent, first a finite binary sequence $Trans(\ell)$ of length $O(\log \ell)$ is defined. We do not need to get into the properties of this sequence. Then the infinite  binary sequence $Tape(\ell)$ is defined as the concatenation of infinitely many copies of $Trans(\ell)$. $Tape(\ell)$ is called  the {\em tape} of the agent with label $\ell$.
The $i$-th copy of $Trans(\ell)$ is called the $i$-th segment of $Tape(\ell)$. 

\item
Ordering of finite paths starting at the node $v$.\\
The following order $(\pi_1,\pi_2,\dots)$ of finite paths starting at a node $v$ is defined:
every path of smaller length precedes every path of larger length, and paths of a given length are ordered lexicographically.
\item
Processing consecutive bits of the tape.\\
For any positive integer $i$, define $a_i=3i^2$.
The $j$-th segment of $Tape(\ell)$ corresponds to the path $\pi_j$. Consecutive bits of $Tape(\ell)$ are processed by allocating time $a_i$ to the $i$-th bit of $Tape(\ell)$ and executing the following actions. If $b_i$ is the $i$-th bit of $Tape(\ell)$ located in the $j$-th segment of $Tape(\ell)$ then: 
\begin{itemize}
\item
if $b_i=1$ then the agent traverses path $\pi_j$, waits $a_i-2|\pi_j|$ rounds at the end of it, and traverses path $rev(\pi_j)$;
\item
if $b_i=0$ then the agent waits $a_i$ rounds.
\end{itemize}
\end{enumerate}
The algorithm is interrupted as soon as the agents meet.
By the design of the algorithm,  the agent starts and ends processing each bit of $Tape(\ell)$ at its starting node $v$.

If agents $A_1$ and $A_2$ start at nodes $v_1$ and $v_2$ respectively, the {\em critical segment} of $A_1$ is defined as follows. Let $\pi$ be the lexicographically smallest among shortest paths from $v_1$ to $v_2$. $\pi$ is called the {\em critical path} of $A_1$. The critical segment of the tape of agent $A_1$ is the segment of its tape corresponding to path $\pi$. The critical segment of the tape of agent $A_2$ is defined similarly. 

It is proved in \cite{BP} that, if agents execute Algorithm {\tt RV}, then there exists an agent $A$ such that rendezvous occurs by the end of the execution of the critical segment of this agent. (The allocation of time periods $a_i=3i^2$ is crucial in the proof). Let $k$ be the the index of the critical path of $A$ in the ordering $(\pi_1,\pi_2,\dots)$. Since each segment contains at most $c\log L$ bits, for some positive constant $c$, the number $N$ of all bits in the tape of $A$ until the end of the critical segment is at most $ck\log L$. This implies the following estimate from \cite{BP} on the time of Algorithm {\tt RV}:
$\sum_{j=1} ^N a_j=\sum_{j=1} ^N 3j^2=\frac{N(N+1)(2N+1)}{2}$.

The key observation permitting to generalize Algorithm {\tt RV} to a universal rendezvous algorithm working for arbitrary connected graphs is that the correctness and complexity of the algorithm does not depend on the ordering of finite paths starting at the node $v$: it depends on the index of the critical path of agent $A$ in the ordering. Of course the ordering  $(\pi_1,\pi_2,\dots)$ from Algorithm {\tt RV} cannot be used if degrees of nodes can be infinite because there are already infinitely many paths of length 1. So, in order to generalize Algorithm {\tt RV} we need to change its part 2.
Recall that Algorithm {\tt UTH} from the previous section yields the following ordering $(\pi_1^*,\pi_2^*,\dots)$ of finite paths starting at the node $v$.

Path $\pi_i^*$ precedes path $\pi_j^*$ if and only if:\\
$\bullet$ $val(\tau(\pi_i^*))< val(\tau(\pi_j^*))$, or\\
$\bullet$ $val(\tau(\pi_i^*))= val(\tau(\pi_j^*))$ and $\tau(\pi_i^*)$ is lexicographically smaller than $\tau(\pi_j^*)$, or \\
$\bullet$ $\tau(\pi_i^*)=\tau(\pi_j^*)$ and $\pi_i^*$ is lexicographically smaller than $\pi_j^*$.

Here is the succinct description of our universal rendezvous algorithm {\tt URV}: 
it is the modification of  Algorithm {\tt RV} from \cite{BP} consisting of replacing the ordering $(\pi_1,\pi_2,\dots)$ of finite paths starting at the node $v$
by the ordering $(\pi_1^*,\pi_2^*,\dots)$ yielded by Algorithm {\tt UTH}. Parts 1. and 3. remain unchanged.
Recall that $W(v_1,v_2)=\max(w(v_1,v_2),w(v_2,v_1))$.

\begin{theorem}
Algorithm {\tt URV} accomplishes rendezvous in arbitrary connected graphs and its time complexity on an instance where agents start at nodes $v_1$, $v_2$ and have labels from the set $\{1,\dots, L\}$, is polynomial in
$\max(W(v_1,v_2),\log L)$ (more precisely, it is $O((W(v_1,v_2) \cdot \log L)^3)$).
\end{theorem}


\begin{proof}
We show how to modify the proof of correctness and complexity of Algorithm {\tt RV} from \cite{BP} to get a similar result for our universal algorithm {\tt URV}.
  Let $A_1$ and $A_2$ be two agents with labels from the set $\{1,\dots, L\}$, starting at nodes $v_1$ and $v_2$, respectively. 
  Let $A$ be the agent, such that rendezvous occurs by the end of the execution of the critical segment of this agent (cf. \cite{BP}).
  Without loss of generality, let $A=A_1$.
  Let $(\pi_1^*,\pi_2^*,\dots)$ be the ordering of finite paths starting at $v_1$,  yielded by Algorithm {\tt UTH}.
  Let $k^*$ be the index of the critical path of $A$ in this ordering. Let $N^*$ be the number of bits of the tape of $A$ processed until the end of the critical segment of $A$.
  
  Similarly as in \cite{BP}, $N^*\leq ck^*\log L$, for some positive constant $c$. 
  By the definition of $(\pi_1^*,\pi_2^*,\dots)$ we have $k^*\leq w(v_1,v_2) \leq W(v_1,v_2)$.
  Following the estimate from \cite{BP} on the time of Algorithm {\tt RV}, we get the following estimate on the time of {\tt URV}:
$\sum_{j=1} ^{N^*} a_j=\sum_{j=1} ^{N^*} 3j^2=\frac{N^*(N^*+1)(2N^*+1)}{2}$.
It follows that the time complexity of our Algorithm {\tt URV} is $O((W(v_1,v_2) \cdot \log L)^3)$ and hence it is polynomial in $\max(W(v_1,v_2) , \log L)$.
  \end{proof}

The following proposition shows that Algorithm {\tt URV} is polynomial in a lower bound on the time of the rendezvous problem.

\begin{proposition}
No universal rendezvous algorithm can have better time complexity than $\Theta(\max(W(v_1,v_2),\log L))$.
\end{proposition}

\begin{proof}
As observed in \cite{BP}, $\Omega (\log L)$ is a lower bound on the time of rendezvous even in the two-node graph. On the other hand,
the order of magnitude $\Theta(W(v_1,v_2))$ cannot be beaten by the execution time of any universal rendezvous algorithm on the tree $T$ all of whose nodes have infinite degree. Indeed, without loss of generality, assume that $w(v_1,v_2) \geq w(v_2,v_1)$.
The adversary does not wake up the agent $A_2$ starting at node $v_2$ until the other agent meets it. Thus agent $A_2$ behaves like the treasure
in the treasure hunt problem.
Hence, the proof is concluded by
Proposition \ref{lb}.
\end{proof}

\section{Conclusion}

We considered treasure hunt and rendezvous for arbitrary connected graphs with possibly countably infinite node degrees. 
We considered anonymous graphs, i.e., graphs without node labels. Of course, adding labels of nodes would not invalidate our positive results, as such labels can be simply ignored by the agents. On the other hand, our lower bounds are still valid for labeled graphs.

While the time of our treasure hunt algorithm has been proved to have optimal order of magnitude, the complexity of the rendezvous algorithm is probably not optimal.
Finding an optimal rendezvous algorithm is a challenging open problem, even for finite graphs.






\end{document}